\documentclass[10pt,a4paper]{article}
\usepackage[latin1]{inputenc}
\usepackage[T1]{fontenc}
\usepackage{amsmath,amscd}
\usepackage{graphicx}
\usepackage[all]{xy}
\usepackage{amsthm}
\usepackage{amssymb,url}
\usepackage{bbm}
\usepackage{ae}
\usepackage[all]{xy}

\newtheorem{sats}{Theorem}[section]
\newtheorem{sats*}{Theorem}

\newtheorem{lem}[sats]{Lemma}
\newtheorem{cor}[sats]{Corollary}
\newtheorem{prop}[sats]{Proposition}

\newcommand{\R}{\mathbbm{R}}
\newcommand{\C}{\mathbbm{C}}

\newcommand{\Z}{\mathbbm{Z}}

\newcommand{\N}{\mathbbm{N}}

\newcommand{\ellL}{\mathcal{L}}

\newcommand{\id}{ \mathrm{id}}
\newcommand{\rd}{\mathrm{d}}

\newcommand{\Ko}{\mathcal{K}}
\newcommand{\Bo}{\mathcal{B}}
\newcommand{\Co}{\mathcal{C}}
\newcommand{\He}{\mathcal{H}}
\newcommand{\T}{\mathbbm{T}}

\newcommand{\Fg}{\mathcal{F}}

\newcommand{\ind}{\mathrm{i} \mathrm{n} \mathrm{d} \,}

\newcommand{\wind}{\mathrm{wi} \mathrm{n} \mathrm{d} \,}

\newcommand{\coker}{\mathrm{c} \mathrm{o} \mathrm{k} \mathrm{e}
\mathrm{r}\,}

\renewcommand{\epsilon}{\varepsilon}
\renewcommand{\phi}{\varphi}

\newcommand{\im}{\mathrm{i} \mathrm{m} \,}

\newcommand{\supp}{\mathrm{s} \mathrm{u} \mathrm{p} \mathrm{p}\,}
\newcommand{\e}{\mathrm{e}}
\newcommand{\tra}{\mathrm{t}\mathrm{r}}

\title{Index formulas and charge deficiencies on the Landau levels}
\author{Magnus Goffeng}
\date{Department of Mathematical Sciences, Division of Mathematics\\
Chalmers university of Technology and University of Gothenburg}

\begin{document}
\maketitle

\bibliographystyle{amsplain}

\begin{abstract}
The notion of charge deficiency is studied from the view of $K$-theory of operator algebras and is applied to the Landau levels in $\R^{2n}$. We calculate the charge deficiencies at the higher Landau levels in $\R^{2n}$ by means of an Atiyah-Singer type index theorem. 
\end{abstract}

\section{Introduction}

The paper is a study of the charge deficiencies at the Landau levels in $\R^{2n}$. The Landau levels are the eigenspaces of the Landau Hamiltonian which is the energy operator for a quantum particle moving in $\R^{2n}$ under the influence of a constant magnetic field of full rank. 

In \cite{avronseilersimon}, the notion of charge deficiency was introduced as a measure off how much a flux tube changes a fermionic system in $\R^2$. The setting of  \cite{avronseilersimon} is a quantum system where the Fermi energy is in a gap and the question is what happens when the system is taken trough a cycle. Letting $P$ denote the projection onto the state space and $U$ the unitary transformation representing the cycle, the projection $Q$ onto the new state space after it had been taken through a cycle can be expressed as $Q=UPU^*$. The relative index $\ind(Q,P)$ is defined as an infinite dimensional analogue of $\dim Q-\dim P$ and is well defined whenever $Q-P$ is a compact operator. The condition that $Q-P$ is compact is equivalent to that $[P,U]$ is compact. In the setting of  \cite{avronseilersimon} the relative index represents the change in the number of fermions that $U$ produces. In \cite{avronseilersimon} the following formula was proven:
\[\ind(Q,P)=\ind(PUP).\]

For sufficiently nice systems in $\R^2$ one can choose the particular unitary given by multiplication by the bounded function $U:=z/|z|$. The condition on the system that is needed is that $P$ commutes with $U$ up to a compact operator. The charge deficiency of a projection $P$ in the sense of  \cite{avronseilersimon} is then defined using $U$ as 
\[c(P):=\ind (PUP).\]

The viewpoint we will have in this paper is that the charge deficiency is a $K$-homology class. This viewpoint lies in line with the view on $D$-brane charges in string theory, see more in \cite{bmrs}, \cite{reisszabo}. In the case studied in \cite{avronseilersimon} the charge deficiency is realized as an odd $K$-homology class on the circle $\T$. The unitary $U$ define a representation of $C(\T)$ and using the fact that $P$ commutes with $U$ up to a compact operator we get a $K$-homology class. Let us denote this $K$-homology class by $[P]$ and by $u$ we will denote the generator of $C(\T)$. In this notation, the charge deficiency is given by $c(P)=[P]\circ [u]\in KK(\C,\C)\cong \Z$, the Kasparov product between $[P]\in K^1(C(\T))$ and $[u]\in K_1(C(\T))$. Thus the charge deficiency is the image of $[P]$ under the isomorphism 
\[K^1(C(\T))=KK_1(C(\T),\C)\cong Hom(K_1(C(\T)),K_0(\C))\cong \Z,\]
where the first isomorphism is the natural mapping coming from the Universal Coefficient Theorem for $KK$-theory and the second isomorphism comes from choosing $[u]$ as a generator for $K_1(C(\T))$. So a better picture is that the $K$-homology class $[P]\in K^1(C(\T))$ is the charge deficiency of $P$.

The system we will consider in this paper consists of a particle moving in $\R^{2n}$ under the influence of a constant magnetic field $B$ of full rank. If we choose a linear vector potential $A$ satisfying $\rd A=B$ the Hamiltonian of this system is given by 
\[H_A:=(-i \nabla-A)^2,\]
This \emph{Landau Hamiltonian} should be viewed as a densely defined operator in the Hilbert space $L^2(\R^{2n})$. Taking $\mathcal{D}(H_A)=C^\infty_c(\R^{2n})$, the operator $H_A$ becomes essentially self-adjoint, see more in \cite{melroz}. Due to the identification $\R^{2n}=\C^n$ we will use the complex structure and we will assume that $B=\frac{i}{2}\sum \rd z_j\wedge \rd \bar{z}_j$.

The Landau Hamiltonian has a discrete spectrum with eigenvalues $\Lambda_\ell=2\ell+n$ for $\ell\in \N$ and the eigenspaces $\ellL^\ell$ are infinite dimensional. Let 
\[P_\ell:L^2(\R^{2n})\to \ellL^\ell\] 
denote the orthogonal projection to the $\ell$:th eigenspace. Our point of view on the charge deficiencies for the Landau levels is that they are $K$-homology classes of the sphere $S^{2n-1}$. For a bounded continuous function $a:\R^{2n}\to M_N(\C)$ we define the continuous function $a_r\in C(S^{2n-1})$ as
\[a_r(v):=  a(rv).\]
We let $A_N$ be the subalgebra of $C_b(\R^{2n})\otimes M_N(\C)$ such that $a_r$ converges uniformly in $v$ to a continuous function $a_\partial$ on $S^{2n-1}$. The mapping $a\mapsto a_\partial$ defines a $*$-homomorphism $A_N\to C(S^{2n-1})\otimes M_N(\C)$. The projection $P_\ell$ commutes up to a compact operator with $a\in A_N$ (see below in Theorem \ref{toeplaqu}) and 
\[P_\ell a|_{\ellL^\ell\otimes \C^N}: \ellL^\ell\otimes \C^N\to \ellL^\ell\otimes \C^N\] 
is Fredholm if and only if $a_\partial$ is invertible (see Proposition \ref{corfred}). Now we may present the main theorem of this paper:

\begin{sats*}
If $a_\partial$ is smooth and invertible, the index of $P_\ell a|_{\ellL^\ell\otimes \C^N}$ can be expressed as 
\[\ind (P_\ell a|_{\ellL^\ell\otimes \C^N})=\frac{-(\ell+n-1)!}{\ell!(2n-1)!(2\pi i)^n}\int _{S^{2n-1}}\tra((a_\partial^{-1}\rd a_\partial)^{2n-1}).\]
The charge deficiency $[P_\ell]\in K^1(C(S^{2n-1}))$ may be expressed in terms of the Bergman projection $P_B$ on the unit ball in $\C^n$ as 
\[[P_\ell]=\frac{(\ell+n-1)!}{\ell !(n-1)!}[P_B].\]
\end{sats*}

\section{The particular Landau levels}

The spectral theory of the Landau Hamiltonian is well known and we will review it briefly. See more in \cite{rozta}. We will let $\phi:=\frac{|z|^2}{4}$ and assume that the magnetic field $B$ is of the form $B=i\partial \bar{\partial} \phi$. Here $\partial$ is the complex linear part of the exterior differential $\rd$. Define the annihilation operators as
\[q_j:=2\frac{\partial}{\partial \bar{z}_j}+z_j \quad \mbox{for}\quad j=1,\ldots,n.\]
The adjoints are given by the creation operators $q_j^*:=-2\frac{\partial} {\partial z_j}+\bar{z}_j$. The annihilation and creation operators satisfy the following formulas:
\begin{align*}
[q_j,q_i]=[q_j^*,q_i^*]=0,& \quad [q_i,q^*_j]=2\delta_{ij}\quad \mbox{and} \quad
H_A=\sum_{j=1}^n q^*_jq_j+&n=\sum_{j=1}^nq_jq^*_j-n.
\end{align*}
Here we view $H_A$ as a densely defined operator in $L^2(\C^{n})$. Thus the lowest eigenvalue is $n$ with corresponding eigenspace $\ellL_0=\e^{-\phi}\Fg(\C^n)$ where $\Fg(\C^n):=L^2(\C^n,\e^{-2\phi})\cap \mathcal{O}(\C^n)$ denotes the Fock space. Here $\mathcal{O}(\C^n)$ denotes the space of holomorphic functions in $\C^n$. In one complex dimension there is only one creation operator $q^*$ and the eigenspaces are given by $\ellL_k=(q^*)^k\ellL_0$. Using multi-index notation, for $\mathbbm{k}=(k_1,\ldots,k_n)\in \N^n$ we define $q_{\mathbbm{k}}:=q_1^{k_1}\cdots q_n^{k_n}$ and 
\[\ellL_{\mathbbm{k}}:=q_{\mathbbm{k}}^*\ellL_0=\ellL_{k_1}\otimes \ellL_{k_2}\otimes \cdots \otimes \ellL_{k_n}.\] 
We will call this space for the \emph{particular} Landau level of height $\mathbbm{k}$. Using that $q_j$ and $q_j^*$ define a representation of the Heisenberg algebra in $n$ dimension we obtain the eigenvalues of $H_A$ as $\Lambda_{\ell}=2\ell+n$ with the corresponding eigenspaces 
\[\ellL^\ell:=\bigoplus_{|\mathbbm{k}|=\ell}\ellL_{\mathbbm{k}}=\bigoplus_{|\mathbbm{k}|=\ell}\ellL_{k_1}\otimes \ellL_{k_2}\otimes \cdots \otimes \ellL_{k_n}.\]
The $\ell$:th eigenspace $\ellL^\ell$ is called the Landau level of height $\ell$. Since the Hamiltonian commutes with the representation of $SU(n)$ on $\C^n$, its eigenspaces are $SU(n)$-invariant. Also the orthogonal projections $P_\ell:L^2(\C^n)\to \ellL^\ell$ are invariant under the $SU(n)$-action.

Recall that the vacuum subspace $\ellL_0\subseteq L^2(\C^n)$ has a reproducing kernel induced by the reproducing kernel on the Fock space. The reproducing kernel of $\Fg(\C^n)$ is given by $K(z,w)=\e^{\frac{w\cdot \bar{z}}{4}}$. So the reproducing kernel of $\ellL_0$ is given by 
\[K_0(z,w):=\e^{\frac{1}{4}(w\cdot \bar{z}-|z|^2-|w|^2)}.\]
This expression for the reproducing kernel implies that the orthogonal projection $P_0:L^2(\C^n)\to \ellL_0$ is given by 
\[P_0f(z)=\int_{\C^n} f(w)\overline{K_0(z,w)}\rd V.\]
By \cite{grigsobolev} the orthogonal projection $P_{\mathbbm{k}}:L^2(\C^n)\to \ellL_{\mathbbm{k}}$ onto the particular Landau level of height $\mathbbm{k}$ is also an integral operator with kernel 
\begin{equation}
\label{projhigher}
K_{\mathbbm{k}}(z,w)=\e^{\frac{1}{4}(w\cdot \bar{z}-|z|^2-|w|^2)}\prod_{j=1}^n L_{k_j}\left(\frac{1}{2}|z_j-w_j|^2\right).
\end{equation}
Here $L_k$ is the Laguerre polynomial of order $k$. Notice that the projections $P_\mathbbm{k}$ are not $SU(n)$-invariant in general.

\section{Toeplitz operators on the Landau levels} 

We want to study topological properties of the particular Landau levels using Toeplitz operators. The symbols will be taken from a suitable subalgebra of $C_b(\C^{n})$, the bounded functions on $\C^n$. The standard notation $\Bo(\He)$ will be used for the bounded operators on a separable Hilbert space $\He$ and the compact operators will be denoted by $\Ko(\He)$. We will let $\pi:C_b(\C^{n})\to \Bo(L^2(\C^n))$ denote the representation given by pointwise multiplication. This is clearly an $SU(n)$-equivariant mapping. Define the linear map $T_\mathbbm{k}:C_b(\C^{n})\to \Bo(\ellL_\mathbbm{k})$ by $T_\mathbbm{k}(a):=P_{\mathbbm{k}}\pi(a)|_{\ellL_\mathbbm{k}}$.

\begin{lem}
\label{complem}
If $a\in C_0(\C^{n})$ then $T_\mathbbm{k}(a)\in \Ko(\ellL_\mathbbm{k})$ for all $\mathbbm{k}\in \N^n$.
\end{lem}

The proof of this lemma is analogous to the proof for the same statement for Toeplitz operators on a pseudoconvex domain from \cite{venugopalkrishna}. 

\begin{proof}
It is sufficient to prove the claim for $a\in C_c(\C^{n})$, since $T_{\mathbbm{k}}$ is continuous and $C_c(\C^{n})\subseteq C_0(\C^{n})$ is dense. Define the compact set $K:=\supp(a)$. Let $R:\ellL_{\mathbbm{k}}\to L^2(\C^n)$ denote the operator given by multiplication by $\chi_{K}$, the characteristic function of $K$. We have $T_{\mathbbm{k}}(a)=P_{\mathbbm{k}}\pi(a)R$ so the Lemma holds if $R$ is compact. That $R$ is compact follows from Cauchy estimates of holomorphic functions on a compact set.
\end{proof}

Define the $SU(n)$-invariant $C^*$-subalgebra $A\subseteq C_b(\C^{n})$ as consisting of functions $a$ such that $a(rv)$ converges uniformly in $v$ as $r\to \infty$ to a continuous function $a_\partial:S^{2n-1}\to \C$ when $r\to \infty$. Thus we obtain a surjective $SU(n)$-equivariant $*$-homomorphism $\pi_\partial:A\to C(S^{2n-1})$ given by 
\[\pi_\partial(a)(v):=  \lim_{r\to \infty}a(rv).\] 
The mapping $\pi_\partial$ satisfies $\ker \pi_\partial=C_0(\C^{n})$. We will henceforth consider $T_{\mathbbm{k}}$ as a mapping from $A$ to $\Bo(\ellL_\mathbbm{k})$.

If we let $B_n$ denote the open unit ball in $\C^n$, another view on $A$ is as the image of the $SU(n)$-equivariant $*$-monomorphism $C(\overline{B_n})\to C_b(B_n)\cong C_b(\C^n)$ where the last isomorphism comes from an $SU(n)$-equivariant homeomorphism $B_n\cong \C^n$.

\begin{sats}
\label{toeplaqu}
The projection $P_{\mathbbm{k}}$ satisfies $[P_{\mathbbm{k}},\pi(a)]\in \Ko(L^2(\C^n))$ for all $a\in A$. Therefore the $*$-linear mapping $T_{\mathbbm{k}}:A\to \Bo(\ellL_\mathbbm{k})$ satisfies 
\[T_{\mathbbm{k}}(ab)- T_{\mathbbm{k}}(a)T_{\mathbbm{k}}(b)\in \Ko(\ellL_\mathbbm{k}).\]
\end{sats}

The proof is based on a similar result from \cite{coburnberger} where the Fock space was used to define a Toeplitz quantization of a certain subalgebra of $L^\infty(\C^n)$. The case of the Fock space is more or less the same as the case $\mathbbm{k}=0$ for Landau quantization. To prove the Theorem we need a lemma similar to part $(iv)$ of Theorem $5$ of \cite{coburnberger}. Using the isomorphism $A\cong C(\overline{B_n})$ we define the   dense subalgebra $A_1\subseteq A$ as the inverse image of the Lipschitz continuous functions in $C(\overline{B_n})$.

\begin{lem}
\label{divlip}
For $a\in A_1$ then for any $\epsilon>0$ we may write $a=g_\epsilon+h_\epsilon$  where $h_\epsilon\in C_0(\C^n)$ and $g_\epsilon\in A$ satisfies
\begin{equation}
\label{felip}
|g_\epsilon(z)-g_\epsilon(w)|\leq \epsilon |z-w| \quad \forall z,w\in \C^n.
\end{equation}
\end{lem}

\begin{proof}
Let $C$ denote the Lipschitz constant of $\pi_\partial(a)$. Take an $\epsilon>0$ and let $\chi_\epsilon$ be a Lipshitz continuous $SU(n)$-invariant cutoff such that $\chi_\epsilon(z)=0$ for $|z|\leq R$ and $\chi_\epsilon(z)=0$ for $|z|\geq 2R$ where $R=R(\epsilon,C)$ is to be defined later. To shorten notation, define $a_\partial:= \pi_\partial(a)$. Let 
\[g_\epsilon(z):=\chi_\epsilon(z)\cdot a_\partial \left(z/|z|\right)\]
and $h_\epsilon:=a-g_\epsilon$. Clearly $h_\epsilon\in C_0(\C^n)$ and $g_\epsilon\in A$ so what remains to be proven is that $R$ can be chosen in such a way that $g_\epsilon$ satisfies equation \eqref{felip}. 

We have elementary estimates 
\[\left| \frac{z}{|z|}-\frac{w}{|w|}\right|\leq \frac{|z-w|}{|z|}+\left| \frac{w}{|z|}-\frac{w}{|w|}\right|\leq 2\frac{|z-w|}{|w|}.\]
Thus for $z,w\neq 0$ the function $a_\partial$ satisfies 
\[\left| a_\partial\left(\frac{z}{|z|}\right)-a_\partial\left(\frac{w}{|w|}\right)\right| \leq \frac{2C}{|w|}|z-w|.\]
The function $\chi_\epsilon$ has Lipschitz coefficient $1/R$ so if we take $R>2C/\epsilon$ then $g_\epsilon$ satisfies equation \eqref{felip}.
\end{proof}

Let $\mathcal{C}(L^2(\C^n)):=\Bo(L^2(C^n))/\Ko(L^2(\C^n))$ denote the Calkin algebra and $\mathfrak{q}$ the quotient mapping.

\begin{proof}[Proof of Theorem \ref{toeplaqu}]
Since Lipschitz continuous functions are dense in $A$ we may assume that $a\in A_1$, so by Lemma \ref{divlip} we can for any $\epsilon>0$ write $a=g_\epsilon+h_\epsilon$. In this case we have for $f\in L^2(\C^n)$
\[[P_\mathbbm{k},\pi(g_\epsilon)]f(z)=\int (g_\epsilon(z)-g_\epsilon(w))K_\mathbbm{k}(z,w)f(w)\rd w.\]
Define the operator 
\[Bf(z):=\int |z-w|K_\mathbbm{k}(z,w)f(w)\rd w.\]
By equation \eqref{projhigher} we have that for some $C$ the integral kernel of $B$ is bounded by 
\[|z-w||K_\mathbbm{k}(z,w)|\leq C|z-w|^{|\mathbbm{k}|+1}\e^{-\frac{1}{8}|z-w|^2}.\] 
Therefore the kernel of $B$ is dominated by the kernel of a bounded convolution operator and $\|B\|<\infty$. The estimate \eqref{felip} for $g_\epsilon$ implies that 
\[\|[P_\mathbbm{k},\pi(g_\epsilon)]\|\leq \epsilon \|B\|.\]
Using that $[P_\mathbbm{k},\pi(g_\epsilon)]=[P_\mathbbm{k},\pi(a)]$ modulo compact operators, by Lemma \ref{complem}, we have the inequality
\[\|\mathfrak{q}([P_\mathbbm{k},a])\|_{\mathcal{C}(L^2(\C^n))}\leq \epsilon \|B\| \quad \forall \epsilon >0.\]
Therefore $\mathfrak{q}([P_\mathbbm{k},a])=0$ and $[P_\mathbbm{k},a]$ is compact.
\end{proof}

Theorem \ref{toeplaqu} implies that the mapping $\tilde{\beta}_\mathbbm{k}:=\mathfrak{q}\circ T_\mathbbm{k}:A\to \mathcal{C}(\ellL_\mathbbm{k})$ is a well defined $*$-homomorphism. Define the $C^*$-algebra 
\[\tilde{\mathcal{T}}_{\mathbbm{k}}:=\{a\oplus x\in A\oplus \Bo(\ellL_\mathbbm{k}):\; \tilde{\beta}_\mathbbm{k}(a)=\mathfrak{q}(x)\}.\]
This $C^*$-algebra contains $\Ko$ as an ideal via the embedding $k\mapsto 0\oplus k$ and we obtain a short exact sequence
\begin{equation}
\label{extdisc}
0\to \Ko\to \tilde{\mathcal{T}}_{\mathbbm{k}}\to A\to 0.
\end{equation}

\begin{lem}
\label{sumlem}
Let $(\mathbbm{k}_p)_{p=1}^N\subseteq \N^n$ be a finite collection of distinct $n$-tuples of integers. Then the mapping 
\[A\ni a\mapsto \mathfrak{q}\left(\left(\sum_{p=1}^N P_{\mathbbm{k}_p}\right)\pi(a)\left(\sum_{p=1}^N P_{\mathbbm{k}_p}\right)\right)\in \Co(\oplus_{p=1}^N \ellL_{\mathbbm{k}_p})\]
coincides with the mapping 
\[A\ni a\mapsto \oplus_{p=1}^N \tilde{\beta}_{\mathbbm{k}_p}(a)\in \Co(\oplus_{p=1}^N \ellL_{\mathbbm{k}_p}).\]
\end{lem}

\begin{proof}
The Lemma follows if we show that $P_\mathbbm{k}\pi(a)P_{\mathbbm{k}'}\in \Ko(L^2(\C^n))$ for $\mathbbm{k}\neq\mathbbm{k}'$. But Theorem  \ref{toeplaqu} implies that $P_\mathbbm{k}\pi(a)(1-P_\mathbbm{k})\in \Ko(L^2(\C^n))$. So the Lemma follows from 
\[P_\mathbbm{k}\pi(a)P_{\mathbbm{k}'}=P_\mathbbm{k}\pi(a)(1-P_\mathbbm{k})P_{\mathbbm{k}'}.\]
\end{proof}

In particular we can look at the collection of all $\mathbbm{k}$:s such that $|\mathbbm{k}|=\ell$. We will define the $SU(n)$-equivariant mapping $\tilde{\beta}_{\ell}:A\to \Co(\ellL^\ell)$ as 
\[a\mapsto \oplus_{|\mathbbm{k}|=\ell}\tilde{\beta}_\mathbbm{k}(a).\] 
Just as for the particular Landau levels we define
\[\tilde{\mathcal{T}}^\ell:=\{a\oplus x\in A\oplus \Bo(\ellL^\ell):\;\tilde{\beta}_\ell(a)=\mathfrak{q}(x)\}.\]
The projection map $\tilde{\mathcal{T}}^\ell\to A$ given by $a\oplus x\mapsto a$ defines an $SU(n)$-equivariant extension 
\[0\to \Ko\to \tilde{\mathcal{T}}^\ell\to A\to 0.\]

\begin{lem}
\label{equikernellem}
The kernel of $\tilde{\beta}_\ell$ is $C_0(\C^n)$.
\end{lem}

\begin{proof}
Lemma \ref{complem} implies that $C_0(\C^n)\subseteq \ker \tilde{\beta}_\ell$. To prove the reverse inclusion we observe that  the mapping $\tilde{\beta}_\ell$ is a unital $SU(n)$-equivariant $*$-homomorphism. Since $\tilde{\beta}_\ell$ is equivariant, the ideal $\ker \tilde{\beta}_\ell\subseteq A$ is $SU(n)$-invariant. The inclusion $C_0(\C^n)\subseteq \ker \tilde{\beta}_\mathbbm{k}$ implies that there is an equivariant surjection $C(S^{2n-1})\to A/\ker\tilde{\beta}_\ell$ which must be an isomorphism since $C(S^{2n-1})$ is $SU(n)$-simple and $\tilde{\beta}_\ell$ is unital. It follows that $\ker\tilde{\beta}_\ell=C_0(\C^n)$.
\end{proof}

It is interesting that although the statement of Lemma \ref{equikernellem} sounds algebraic, it is really the analytic statement that $T_\ell(a)$ is compact if and only if $a$ vanishes at infinity. And this is proven with algebraic methods!

\begin{prop}
\label{corfred}
If $u\in A\otimes M_N$, the operator $T_{\ell}(u)$ is Fredholm if and only if $\pi_\partial(u)$ is invertible.
\end{prop}

\begin{proof}
By Atkinson's Theorem  $T_{\ell}(u)$ is Fredholm if and only if $\tilde{\beta}_\ell(u)$ is invertible. Lemma \ref{equikernellem} implies that $\ker \pi_\partial=\ker \tilde{\beta}_\ell$ so $\tilde{\beta}_\ell(u)$ is invertible if and only if $\pi_\partial (u)$ is invertible.
\end{proof}

\section{Pulling symbols back from $S^{2n-1}$}

To put the Toeplitz operators on a Landau level in a suitable homological picture, we must pass from $A$ to $C(S^{2n-1})$. This is a consequence of the circumstance that $A$ is homotopy equivalent to $\C$, so $A$ does not contain any relevant topological information. With Lemma \ref{equikernellem} in mind we define the Toeplitz algebra $\mathcal{T}_\mathbbm{k}$ for $C(S^{2n-1})$ as if $\beta_\mathbbm{k}$ were injective. So let $\lambda:C(S^{2n-1})\to \Bo(L^2(\C^n))$ denote the $*$-representation defined by 
\begin{equation}
\label{deflam}
\lambda(a)f(z)=a\left(\frac{z}{|z|}\right)f(z).
\end{equation}
Take $\chi_0\in C^\infty (\R)$ to be a smooth function such that $\chi_0(x)=0$ for $|x|\leq 1$ and $1-\chi_0\in C_c^\infty(\R)$. We define the cut-off $\chi(z):=\chi_0(|z|)$ and the operator 
\begin{equation}
\label{almostproj}
\tilde{P}_\mathbbm{k}:=P_\mathbbm{k}\chi. 
\end{equation} 
For the operator $\tilde{P}_\mathbbm{k}$, $\mathfrak{q}(\tilde{P}_\mathbbm{k})$ is a projection by Lemma \ref{complem}. We let $\mathcal{T}_\mathbbm{k}$ be the $C^*$-algebra generated by $\tilde{P}_\mathbbm{k}\lambda(C(S^{2n-1}))\tilde{P}_\mathbbm{k}^*$.

\begin{sats}
\label{genlam}
For any $\mathbbm{k},\mathbbm{k}'\in \N^n$ there exist a unitary 
\[Q_{\mathbbm{k},\mathbbm{k}'}:\ellL_{\mathbbm{k}'}\to \ellL_\mathbbm{k}\]
such that $Ad(Q_{\mathbbm{k},\mathbbm{k}'}):\mathcal{T}_\mathbbm{k}\to \mathcal{T}_{\mathbbm{k}'}$ is an isomorphism satisfying 
\begin{equation}
\label{step}
\mathfrak{q}(\tilde{P}_{\mathbbm{k}'}\lambda(a)\tilde{P}_{\mathbbm{k}'}^*)=\mathfrak{q}\circ Ad(Q_{\mathbbm{k},\mathbbm{k}'})(\tilde{P}_{\mathbbm{k}}\lambda(a)\tilde{P}_{\mathbbm{k}}^*). 
\end{equation}
Furthermore, for any $\mathbbm{k}\in \N^n$, the representation of $\mathcal{T}_\mathbbm{k}$ on $\ellL_\mathbbm{k}$ given by the inclusion $\mathcal{T}_\mathbbm{k}\subseteq \Bo(\ellL_\mathbbm{k})$ is irreducible and has the cyclic vector $\xi_\mathbbm{k}$ defined by 
\[\xi_{\mathbbm{k}}(z):=q^*_\mathbbm{k}(\e^{-|z|^2/4}).\] 
Up to normalization the cyclic vectors satisfy 
\[Q_{\mathbbm{k},\mathbbm{k}'}\xi_{\mathbbm{k}'}=\xi_\mathbbm{k}.\]
\end{sats}

\begin{proof}
Let us start with observing that for any $a,b\in C(S^{2n-1})$ we have 
\[\tilde{P}_\mathbbm{k}\lambda(ab)\tilde{P}_\mathbbm{k}^*-\tilde{P}_\mathbbm{k}\lambda(a)\tilde{P}_\mathbbm{k}^*P_\mathbbm{k}\lambda(b)\tilde{P}_\mathbbm{k}^*\in \Ko.\]
So if $\mathcal{T}_\mathbbm{k}$ acts irreducibly on $\ellL_\mathbbm{k}$, then $\Ko\subseteq \mathcal{T}_\mathbbm{k}$.

First we will construct a cyclic vector for the $\mathcal{T}_\mathbbm{k}$-action on $\ellL_\mathbbm{k}$ and use the cyclic vector in $\ellL_0$ to show that $\mathcal{T}_0$ acts irreducibly on $\ellL_0$. Then we will show that for $\mathbbm{k}$ such that $\mathcal{T}_\mathbbm{k}$ acts irreducibly on $\ellL_\mathbbm{k}$ and $1\leq j\leq n$ there is an isomorphism $\mathcal{T}_\mathbbm{k}\cong \mathcal{T}_{\mathbbm{k}+e_j}$ induced by a unitary intertwining the $\mathcal{T}_\mathbbm{k}$-action on $L_{\mathbbm{k}}$ with the $\mathcal{T}_{\mathbbm{k}+e_j}$-action on $\ellL_{\mathbbm{k}+e_j}$. 

Consider the elements $\xi_{\mathbbm{m},\mathbbm{k}}\in \ellL_\mathbbm{k}$ for $\mathbbm{m}\in \N^n$ defined by
\[\xi_{\mathbbm{m},\mathbbm{k}}(z):=q^*_\mathbbm{k}(z^\mathbbm{m}\e^{-|z|^2/4}).\]
The elements $\xi_{\mathbbm{m},\mathbbm{k}}$ form an orthogonal basis for $\ellL_\mathbbm{k}$. As in the statement of the theorem, we define $\xi_\mathbbm{k}:=\xi_{0,\mathbbm{k}}$. For $a\in C(S^{2n-1})$ we have
\begin{align*}
\langle \xi_{\mathbbm{m},\mathbbm{k}},\tilde{P}_\mathbbm{k}a\tilde{P}_\mathbbm{k}^*\xi_\mathbbm{k}\rangle =\langle \xi_{\mathbbm{m},\mathbbm{k}},\chi^2 a\xi_\mathbbm{k}\rangle&=\\
\int _{\C^n}  \bar{q}^*_\mathbbm{k}(\bar{z}^\mathbbm{m}\e^{-|z|^2/4})q^*_\mathbbm{k}(\e^{-|z|^2/4})\chi^2(z)a\left(\frac{z}{|z|}\right)\rd V&=\int _{S^{2n-1}}p_{\mathbbm{m}}(\bar{z})a(z)\rd S,
\end{align*}
for some polynomials $p_\mathbbm{m}$ of degree at most $2|\mathbbm{k}|+|\mathbbm{m}|$. It follows that $\mathcal{T}_\mathbbm{k}\xi_\mathbbm{k}$ span $\ellL_\mathbbm{k}$ and therefore $\overline{\mathcal{T}_\mathbbm{k}\xi_\mathbbm{k}}=\ellL_\mathbbm{k}$. Thus $\xi_\mathbbm{k}$ is a cyclic vector for the $\mathcal{T}_\mathbbm{k}$-action. 

By standard theory $\mathcal{T}_0$ acts irreducibly on $\ellL_0$ if and only if there are no non-zero $\xi_0',\xi_0''\in \ellL_0$ such that $\xi_0=\xi_0'+\xi_0''$ and $\mathcal{T}_0\xi_0'\perp \mathcal{T}_0\xi_0''$. Assume that for some $\xi_0'\in \ellL_0$ we have $\mathcal{T}_0\xi_0'\perp \mathcal{T}_0(\xi_0-\xi_0')$. The orthogonality condition implies that $\langle \tilde{P}_0a\tilde{P}_0^*(\xi_0-\xi_0'),\xi_0'\rangle=0$ for all $a\in C(S^{2n-1})$ and $P_0$ is self-adjoint so this relation is equivalent to $\langle \chi^2a\xi_0,\xi_0'\rangle=\langle \chi^2a\xi_0',\xi_0'\rangle$ for all $a\in C(S^{2n-1})$. There exist a holomorphic function $f_0$ such that $\xi_0'(z)=f_0(z)\e^{-|z|^2/4}$ and the equation $\langle \chi^2a\xi_0,\xi_0'\rangle=\langle \chi^2a\xi_0',\xi_0'\rangle$ implies
\[\int_{\C^n}\overline{f_0(z)}\e^{-|z|^2/2}\chi^2(z)a\left(\frac{z}{|z|}\right)\rd V=\int_{\C^n}|f_0(z)|^2\e^{-|z|^2/2}\chi^2(z)a\left(\frac{z}{|z|}\right)\rd V.\]
Hence $f_0$ must be real, and since it is holomorphic it must be constant. Thus $\xi_0'$ is in the linear span of $\xi_0$ and $\xi_0$ defines a pure state. Since the $\mathcal{T}_0$-action on $\ellL_0$ has a pure state, it is irreducible. 

Assume that $\mathcal{T}_\mathbbm{k}$ acts irreducibly on $\ellL_\mathbbm{k}$. Consider the polar decomposition of the unbounded operator $q_j$ on $L^2(\C^n)$, that is $q_j^*=E_jQ_j$ where $Q_j$ is a coisometry and $E_j$ is a strictly positive unbounded operator. Clearly $E_j$ is diagonal on the energy levels and
\[E_j=\bigoplus _{\mathbbm{k}'\in \N^n} \sqrt{k'_j}\;P_{\mathbbm{k}'}.\]
We define the $*$-homomorphism $\rho_j:\mathcal{T}_{\mathbbm{k}+e_j}\to \Bo(\ellL_\mathbbm{k})$ by $\rho_j(T):= Q_j^*TQ_j|_{\ellL_\mathbbm{k}}$. Since $Q_j$ is a coisometry this is clearly a $*$-monomorphism. It follows from the fact that $q_j^*|:\ellL_\mathbbm{k}\to \ellL_{\mathbbm{k}+e_j}$ is an isomorphism, that $Q_j|:\ellL_\mathbbm{k}\to \ellL_{\mathbbm{k}+e_j}$ is unitary, so $\rho_j$ is unital. If $a\in C^\infty(S^{2n-1})$ then for some non-zero constant $c$ we have
\begin{align*}
\rho_j(\tilde{P}_{\mathbbm{k}+e_j}\lambda(a)\tilde{P}_{\mathbbm{k}+e_j}^*)=c q_j\tilde{P}_{\mathbbm{k}+e_j}\lambda(a)\tilde{P}_{\mathbbm{k}+e_j} ^*q_j^*|_{\ellL_\mathbbm{k}}&=\\
=cP_\mathbbm{k}\left[\frac{\partial}{\partial \bar{z}_j},\chi^2\lambda(a)\right]P_{\mathbbm{k}+e_j}q_j^*|_{\ellL_\mathbbm{k}}&+\tilde{P}_{\mathbbm{k}}\lambda(a)\tilde{P}_{\mathbbm{k}}^*\in \mathcal{T}_\mathbbm{k},
\end{align*}
because Theorem \ref{toeplaqu} implies $P_\mathbbm{k}bP_{\mathbbm{k}+e_j}\in \Ko(L^2(\C^n))$ for $b\in A$ and by the induction assumption $\Ko\subseteq \mathcal{T}_\mathbbm{k}$. So we obtain a $*$-monomorphism $\rho_j: \mathcal{T}_{\mathbbm{k}+e_j}\to \mathcal{T}_{\mathbbm{k}}$. However, we have cyclic vectors $\xi_\mathbbm{k}$ and $\xi_{\mathbbm{k}+e_j}$ for $\mathcal{T}_\mathbbm{k}$ respectively $\mathcal{T}_{\mathbbm{k}+e_j}$. For these vectors, $Q_j\xi_\mathbbm{k}$ is a multiple of $\xi_{\mathbbm{k}+e_j}$ so 
\[\ellL_{\mathbbm{k}+e_j}=\overline{\mathcal{T}_{\mathbbm{k}+e_j}\xi_{\mathbbm{k}+e_j}}\xrightarrow{Q_j^*}\overline{\mathcal{T}_{\mathbbm{k}}\xi_\mathbbm{k}}.\]
Therefore $\rho_j$ is surjective and an isomorphism. We conclude that $\mathcal{T}_\mathbbm{k}$ is independent of $\mathbbm{k}$ and the representations on $\ellL_\mathbbm{k}$ are irreducible since $\xi_0$ is pure and the $\mathcal{T}_\mathbbm{k}$-actions are all equivalent.
\end{proof}

In \cite{bpr} a weaker, but more explicit, statement was proven in complex dimension $1$. Lemma $9.2$ of \cite{bpr} gives an explicit expression of $Q_{k,0}^*T_k(a)Q_{k,0}$ if $a\in A$ is smooth as 
\[Q_{k,0}^*T_k(a)Q_{k,0}=T_0(\mathcal{D}_k(a)),\]
where $\mathcal{D}_k:=\id +\sum _{j=1}^k d_{j,k}\Delta^{j}$, for some explicit constants $d_{j,k}$ and $\Delta$ is the Laplacian on $\C$.  \\

For $i=1,\ldots, n$ we let $z_i:S^{2n-1}\to \C$ denote the coordinate functions of the embedding $S^{2n-1}\subseteq \C^n$. Clearly $z_i\in C(S^{2n-1})$.

\begin{cor}
\label{genlem}
The operators $P_\mathbbm{k}\lambda(z_i)P_\mathbbm{k}^*$ together with $\Ko$ generate $\mathcal{T}_\mathbbm{k}$ as a $C^*$-algebra.
\end{cor}

\begin{proof}
Let $U$ denote the $C^*$-algebra generated by $P_\mathbbm{k}\lambda(z_i)P_\mathbbm{k}$ and $\Ko$. The $C^*$-algebra $\mathcal{T}_\mathbbm{k}$ is  constructed as the $C^*$-algebra generated by the linear space $P_\mathbbm{k}\lambda(C(S^{2n-1}))P_\mathbbm{k}$ because $P_\mathbbm{k}\lambda(a)P_\mathbbm{k}-\tilde{P}_\mathbbm{k}\lambda(a)\tilde{P}^*_\mathbbm{k}\in \Ko$. So it is sufficient to prove  $P_\mathbbm{k}\lambda(C(S^{2n-1}))P_\mathbbm{k}\subseteq U$. Given a function $a\in C(S^{2n-1})$ the Stone-Weierstrass theorem implies that there is a sequence of polynomials $R_j=R_j(z,\bar{z})$ such that $R_j\to a$ in $C(S^{2n-1})$. The functions $R_j$ are polynomials so it follows that  
\[P_\mathbbm{k}\lambda(R_j)P_\mathbbm{k}-R_j(P_\mathbbm{k}\lambda(z)P_\mathbbm{k}, P_\mathbbm{k}\lambda(z^*)P_\mathbbm{k})\in \Ko\]
 and $P_\mathbbm{k}\lambda(R_j)P_\mathbbm{k}\in U$. Finally $\|P_\mathbbm{k}\lambda(R_j)P_\mathbbm{k}-P_\mathbbm{k}\lambda(a)P_\mathbbm{k}\|_{\Bo(\ellL_\mathbbm{k})}\leq \|R_j-a\|_{C(S^{2n-1})}$ which implies $P_\mathbbm{k}\lambda(a)P_\mathbbm{k}\in U$.
\end{proof}

\begin{cor}
\label{kernellem}
The mapping $\beta_\mathbbm{k}:C(S^{2n-1})\to \Co(\ellL_\mathbbm{k})$ induced from $\tilde{\beta}_\mathbbm{k}$ is injective, so if $u\in A\otimes M_N$ the operator $T_{\mathbbm{k}}(u)$ is Fredholm if and only if $\pi_\partial(u)$ is invertible.
\end{cor}

\begin{proof}
Due to equation \eqref{step} in Theorem \ref{genlam}, the Corollary follows from Lemma \ref{equikernellem}. The proof of the second statement of the Corollary is proven in the same fashion as Proposition \ref{corfred}.
\end{proof}

From the fact that the mapping $\beta_\mathbbm{k}$ is injective it follows that the symbol mapping $\tilde{P}_\mathbbm{k}\lambda(a)\tilde{P}_\mathbbm{k}^*\mapsto a$ gives a well defined surjection $\sigma_\mathbbm{k}:\mathcal{T}_\mathbbm{k}\to C(S^{2n-1})$. Clearly the kernel of $\sigma_\mathbbm{k}$ is non-zero and $\ker \sigma_\mathbbm{k}\subseteq \Ko$, so by Theorem \ref{genlam} $\ker \sigma_\mathbbm{k}= \Ko$. Therefore we can construct the exact sequence
\begin{equation}
\label{lato}
0\to \Ko\to \mathcal{T}_\mathbbm{k}\xrightarrow{\sigma_\mathbbm{k}} C(S^{2n-1})\to 0.
\end{equation}
A completely positive splitting of the symbol mapping $\sigma_\mathbbm{k}:\mathcal{T}_\mathbbm{k}\to C(S^{2n-1})$ is given by $a\mapsto \tilde{P}_\mathbbm{k}\lambda(a)\tilde{P}_\mathbbm{k}^*$. 

The exact sequence \eqref{lato} defines an extension class $[\mathcal{T}_\mathbbm{k}]\in Ext(C(S^{2n-1}))$. To read more about $Ext$, $K$-theory and $K$-homology we refer the reader to the references. Since $C(S^{2n-1})$ is a nuclear $C^*$-algebra there is an isomorphism $Ext(C(S^{2n-1}))\cong K^1(C(S^{2n-1}))$ and we can describe the $K$-homology class of $[\mathcal{T}_\mathbbm{k}]$ explicitely by a Fredholm module as follows; we let $\lambda:C(S^{2n-1})\to \Bo(L^2(\C^n))$ be as in equation \eqref{deflam} and define the operator 
\[F_\mathbbm{k}=\frac{(1+\tilde{P}_\mathbbm{k})}{2}\] 
where $\tilde{P}_\mathbbm{k}$ is as in equation \eqref{almostproj}. Clearly, $(L^2(\C^n),\lambda,F_\mathbbm{k})$ defines a Fredholm module which represents the image of $[\mathcal{T}_\mathbbm{k}]$ in $K^1(C(S^{2n-1}))$.

\begin{cor}
\label{indepland}
The class $[\mathcal{T}_\mathbbm{k}]\in Ext(C(S^{2n-1}))$ is independent of $\mathbbm{k}$.
\end{cor}

\begin{proof}
The extension $\mathcal{T}_\mathbbm{k}$ is equivalent to $\mathcal{T}_{\mathbbm{k}'}$ since it follows from equation \eqref{step} that the following diagram with exact rows commute
\[
\begin{CD}
0@>>>   \Ko@>>>\mathcal{T}_{\mathbbm{k}'}@>>> C(S^{2n-1})@>>>0\\
@. @VVAd(Q_{\mathbbm{k},\mathbbm{k}'})V @VV Ad(Q_{\mathbbm{k},\mathbbm{k}'})V @|@.\\
0@>>>   \Ko@>>>\mathcal{T}_\mathbbm{k}@>>> C(S^{2n-1})@>>>0\\
\end{CD}.\]
\end{proof}

So we know that $[\mathcal{T}_\mathbbm{k}]$ is independent of $\mathbbm{k}$, this implies that the index of $T_\mathbbm{k}(u)$ for $u\in M_n\otimes A$ is independent of $\mathbbm{k}$. But how do we calculate it? The index theorem that allows the calculation involves studying how the coordinate functions on $S^{2n-1}$ act on the monomial base of $\ellL_0$. We will first review some theory of Toeplitz operators on the Bergman space and then study what happens in complex dimension $1$ and $2$.\\

The Bergman space on the unit ball $B_n\subseteq \C^n$ is defined as $A^2(B_n):=L^2(B_n)\cap \mathcal{O}(B_n)$, that is; holomorphic functions on $B_n$ which are square integrable. The Bergman space is a closed subspace of $L^2(B_n)$ and we will denote the orthogonal projection $L^2(B_n)\to A^2(B_n)$ by $P_B$. 

The Bergman projection defines a $K$-homology class $[P_B]\in K^1(C(S^{2n-1}))$ in the same fashion as for the Landau projections. That is, for $a\in C(\overline{B_n})$ the operator  $[P_B,a]\in \Bo(L^2(B_n))$ is compact. The reason that we can use $P_B$ to define a $K$-homology class for $S^{2n-1}$ instead of $\overline{B_n}$ is analogously to above that  $P_B a|_{A^2(B_n)}$ is compact if and only if $a\in C_0(B_n)$, see more in \cite{venugopalkrishna}. Thus $P_B a|_{A^2(B_n)}$ is Fredholm if and only if $a|_{S^{2n-1}}$ is invertible. 

Furthermore $[P_B,a]$ is compact. So $[P_B]$ is a well defined $K$-homology class in $K^1(C(S^{2n-1}))$. By \cite{boutetdemonvel} the following index formula holds for the Toeplitz operator $P_B a|_{A^2(B_n)}$ if the symbol $a_\partial :=a|_{S^{2n-1}}$ is smooth:
\begin{equation}
\label{indform}
\ind (P_B a|_{A^2(B_n)})=\frac{-(n-1)!}{(2n-1)!(2\pi i)^n}\int _{S^{2n-1}}\tra((a_\partial^{-1}\rd a_\partial)^{2n-1}).
\end{equation}
This formula was also proven in \cite{guehig} by an elegant use of Atiyah-Singers index theorem. 

We will by $\mathcal{T}^n$ denote the $C^*$-algebra generated by $P_B C(\overline{B_n}) P_B$ in $\Bo(A^2(B_n))$. The $K$-homology class $[P_B]\in K^1(C(S^{2n-1}))$ can be represented by the extension class $[\mathcal{T}^n]\in Ext(C(S^{2n-1}))$ defined by means of the short exact sequence 
\begin{equation}
\label{to}
0\to \Ko\to \mathcal{T}^n\xrightarrow{\sigma^n} C(S^{2n-1})\to 0.
\end{equation}

\section{The special cases $\C$ and $\C^2$}

In this chapter we will study the special cases of complex dimension $1$ and $2$. Dimension $1$ has been studied previously in \cite{avronseilersimon} and provides a simpler picture than in higher dimensions. In the $1$-dimensional case we have that $K_1(C(\T))\cong \Z$ and we can take the coordinate function $z:\T\to \C$ to be a generator. So when we want to determine the class $[\mathcal{T}_k]$ we only need to calculate the index of $P_k\lambda(z)P_k$ where $\lambda$ is as in equation \eqref{deflam}. We recall the following Proposition from \cite{avronseilersimon}:

\begin{prop}[Proposition $7.3$ from \cite{avronseilersimon}]
For any $k\in \N$ we have that  
\[\ind(P_k\lambda(z)P_k)=-1.\]
\end{prop}

The method used in \cite{avronseilersimon} to prove this Proposition was to show that in a suitable basis $P_k\lambda(z)P_k$ was up to some coefficients a unilateral shift. In higher dimension the proof is based on similar ideas.

\begin{sats}
\label{assthm}
For $n=1$ there is an isomorphism $\mathcal{T}_k\cong \mathcal{T}^1$  making $[\mathcal{T}_k]= [\mathcal{T}^1]\in K^1(C(\T))$.
\end{sats}

\begin{proof}
By Proposition $7.3$ of \cite{avronseilersimon} 
\begin{equation}
\label{onedeq}
[\mathcal{T}_k].[u]=\ind(P_k \lambda(u) P_k)=-\wind (u)=[\mathcal{T}^1].[u]
\end{equation}
for an invertible function $u\in C(\T)$. Here $\wind(u)$ denotes the winding number of $u$ which is defined for smooth $u$ as 
\[\wind(u):=\frac{1}{2\pi i}\int _\T u^{-1}\rd u\]
and defines an isomorphism $K_1(C(\T))\to \Z$. By the Universal Coefficient Theorem for $KK$-theory (see Theorem $4.2$ of \cite{rosenbergschochet}) the mapping 
\[K^1(C(\T))\to Hom(K_1(C(\T)),\Z)\] 
is an isomorphism so equation \eqref{onedeq} implies that $[\mathcal{T}_k]= [\mathcal{T}^1]$. 

By Theorem $13$ of \cite{doug}, the short exact sequence $0\to \Ko\to \mathcal{T}_k\to C(\T)\to 0$ is characterized by an isometry $v$ such that $vv^*-1$ is compact and $\mathcal{T}_k$ is generated by $v$. Then $z\mapsto v$ defines a splitting and the symbol mapping $\mathcal{T}_k\to C(\T)$ is just $v\mapsto z$. By equation \eqref{onedeq}, $1-vv^*$ is a rank one projection, so the theorem follows.
\end{proof}

Also in dimension $2$ we can find a generator for the odd $K$-theory. As generator for $K_1(C(S^3))\cong \Z$ we can take the diffeomorphism $u:S^3\to SU(2)$ defined as
\[u(z_1,z_2):=\begin{pmatrix} z_1&z_2\\
-\bar{z}_2& \bar{z}_1\end{pmatrix}.\]

\begin{prop}
\label{tl}
The extension class $[\mathcal{T}^2]$ generate $K^1(C(S^3))$ and $[u]$ generate $K_1(C(S^3))$.
\end{prop} 

\begin{proof} Recalling that $P_B$ denotes the Bergman projection we will start by calculating the index of the Toeplitz operator $P_BuP_B:A^2(B_2)\otimes \C^2\to A^2(B_2)\otimes \C^2$. Using the index theorem by Boutet de Monvel (\cite{boutetdemonvel}) reviewed above in equation \eqref{indform}, the following index formula holds for smooth $u$:
\begin{equation}
\label{td}
\ind(P_BuP_B)=-\frac{1}{3!(2\pi i)^2}\int _{S^{3}}\tra((u^*\rd u)^{3}).
\end{equation}
A straightforward calculation gives that 
\[\tra((u^*\rd u)^{3})=3(z_1\rd \bar{z}_1-\bar{z}_1\rd z_1)\wedge \rd z_2\wedge \rd \bar{z}_2+3(z_2\rd\bar{z}_2-\bar{z}_2\rd z_2)\wedge \rd z_1\wedge \rd \bar{z}_1.\]
Invoking Stokes Theorem on equation \eqref{td} gives that 
\begin{align*}
-\frac{1}{3!(2\pi i)^2}\int _{S^{3}}\tra((u^*\rd u)^{3})=\frac{1}{48\cdot vol(B_2)}\int _{B_2}\rd\tra((u^*\rd u)^{3})=&\\
=\frac{1}{4\cdot vol(B_2)}\int _{B_2}\rd z_1\wedge \rd \bar{z}_1\wedge \rd z_2\wedge \rd \bar{z}_2=-\frac{1}{vol(B_2)}\int _{B_2}\rd V=&-1
\end{align*}
This equation shows that 
\begin{equation}
\label{indtdtoep}
[\mathcal{T}^2].[u]=\ind(P_BuP_B)=-1.
\end{equation}

Consider the split-exact sequence $0\to C_0(\R^3)\to C(S^3)\to \C\to 0$ where the mapping $C(S^3)\to \C$ is point evaluation. Since the sequence splits, and $K_1(\C)=K^1(\C)=0$ the embedding $C_0(\R^3)\to C(S^3)$ induces isomorphisms $K_1(C(S^3))\cong K_1(C_0(\R^3))=\Z$ and $K^1(C(S^3))\cong K^1(C_0(\R^3))=\Z$. So the Kasparov product $K_1(C(S^3))\times K^1(C(S^3))\to \Z$ is just a pairing $\Z\times \Z\to \Z$, and since $[\mathcal{T}^2].[u]=-1$ it follows that  $[\mathcal{T}^2]$ generates $K^1(C(S^3))$ and $[u]$ generates $K_1(C(S^3))$.
\end{proof}

\begin{sats}
For any $\mathbbm{k}\in \N^2$ we have
\begin{equation}
\label{indtd}
\ind (P_\mathbbm{k}\lambda(u)P_\mathbbm{k})=-1.
\end{equation}
Therefore $[\mathcal{T}^2]=[\mathcal{T}_\mathbbm{k}]$.
\end{sats}

\begin{proof}
If equation \eqref{indtd} holds, $[\mathcal{T}^2]=[\mathcal{T}_\mathbbm{k}]$ follows directly from equation \eqref{indtdtoep} using the Universal Coefficient Theorem for $KK$-theory (see Theorem $4.2$ of \cite{rosenbergschochet}). This is a consequence of the fact that the natural mapping 
\[K^1(C(S^3))\to Hom(K_1(C(S^3)),\Z)\] 
is an isomorphism. The injectivity of this map implies that if $[\mathcal{T}^2].[u]=[\mathcal{T}_\mathbbm{k}].[u]$ for a generator $[u]$ then $[\mathcal{T}^2]=[\mathcal{T}_\mathbbm{k}]$.

To prove equation \eqref{indtd} we take $\mathbbm{k}=0$, since Corollary \ref{indepland} implies that the integer $\ind (P_\mathbbm{k}\lambda(u)P_\mathbbm{k})$ is independent of $\mathbbm{k}$. We claim that $P_0\lambda(u)P_0$ is an injective operator and the cokernel of $P_0\lambda(u)P_0$ is spanned by the $\C^2$-valued function $z\mapsto \e^{-|z|^2/4}\oplus 0$. This statement will prove the theorem.

To prove that $P_0\lambda(u)P_0$ is injective, assume $f\in \ker (P_0\lambda(u)P_0)$. Define the functions 
\[\xi^\mathbbm{m}(z):= z^\mathbbm{m}\e^{-|z|^2/4}\]
for $\mathbbm{m}\in \N^2$. The functions $\xi^\mathbbm{m}$ form an orthogonal basis for $\ellL_0$ by Theorem $1.63$ of \cite{follandphase}. Expand the function $f$ in an $L^2$-convergent series 
\[f=\sum_{\mathbbm{m}\in \N^2} c_\mathbbm{m}\xi^\mathbbm{m},\]
where $c_\mathbbm{m}=c_\mathbbm{m}^{(1)}\oplus c_\mathbbm{m}^{(2)}\in \C^2$. Since $f\in \ker (P_0\lambda(u)P_0)$ we have the following orthogonality condition
\[0=\langle \xi^{\mathbbm{m}'}\oplus 0,\lambda(u)f\rangle=\sum _\mathbbm{m} \int_{\C^2} \left( c_\mathbbm{m}^{(1)}\frac{\bar{z}^{\mathbbm{m'}}z^{\mathbbm{m}+e_1}}{|z|}+c_\mathbbm{m}^{(2)}\frac{\bar{z}^{\mathbbm{m'}}z^{\mathbbm{m}+e_2}}{|z|}\right)\e^{|z|^2/2}\rd V=\]
\[=\sum_\mathbbm{m} t_{\mathbbm{m},\mathbbm{m}'}\int_{S^3} \left( c_\mathbbm{m}^{(1)}\bar{z}^{\mathbbm{m'}}z^{\mathbbm{m}+e_1}+c_\mathbbm{m}^{(2)}\bar{z}^{\mathbbm{m'}}z^{\mathbbm{m}+e_2}\right)\rd S,\]
for some coefficients $t_{\mathbbm{m},\mathbbm{m}'}$, for a detailed calculation of $t_{\mathbbm{m},\mathbbm{m}'}$ see below in Proposition \ref{firstisometry}. Using that the functions $\xi^\mathbbm{m}$ are orthogonal we obtain that there exist a $C_\mathbbm{m}>0$ such that 
\begin{equation}
\label{coe}
c_{\mathbbm{m}-e_1}^{(1)}=-C_\mathbbm{m}c_{\mathbbm{m}-e_2}^{(2)}.
\end{equation}
On the other hand, we have 
\[0=\langle 0\oplus  \xi^{\mathbbm{m}'},\lambda(u)f\rangle=\sum _\mathbbm{m} \int_{\C^2} \left( -c_\mathbbm{m}^{(1)}\frac{\bar{z}^{\mathbbm{m'}+e_2}z^{\mathbbm{m}}}{|z|}+c_\mathbbm{m}^{(2)}\frac{\bar{z}^{\mathbbm{m'}+e_1}z^{\mathbbm{m}}}{|z|}\right)\e^{|z|^2/2}\rd V=\]
\[=\sum_\mathbbm{m}t_{\mathbbm{m},\mathbbm{m}'}\int_{S^3} \left( -c_\mathbbm{m}^{(1)}\bar{z}^{\mathbbm{m'}+e_2}z^{\mathbbm{m}}+c_\mathbbm{m}^{(2)}\bar{z}^{\mathbbm{m'}+e_1}z^{\mathbbm{m}}\right)\rd S.\]
Again using orthogonality of the functions $\xi^\mathbbm{m}$ we obtain that there is a $C'_\mathbbm{m}>0$ such that 
\begin{equation}
\label{coetva}
c_{\mathbbm{m}+e_2}^{(1)}=C_\mathbbm{m}'c_{\mathbbm{m}+e_1}^{(2)}.
\end{equation}
Equation \eqref{coe} implies $c^{(1)}_\mathbbm{m}=0$ for $m_2=0$. For $m_2>0$ equation \eqref{coe}  implies 
\[c_{\mathbbm{m}}^{(1)}=-C_{\mathbbm{m}+e_1}c_{\mathbbm{m}-e_2+e_1}^{(2)}.\] 
Then equation \eqref{coetva} for $\mathbbm{m}-e_2$ gives 
\[c^{(1)}_\mathbbm{m}\left(1+\frac{C_{\mathbbm{m}+e_1}}{C_{\mathbbm{m}-e_2}'}\right)=0.\]
So $c^{(1)}_\mathbbm{m}=0$ for all $\mathbbm{m}$. Equation \eqref{coe} implies $c^{(2)}_\mathbbm{m}=0$ for all $\mathbbm{m}$. Thus $f=0$ and $\ker(P_0\lambda(u)P_0)=0$.

The second statement, that the cokernel of $P_0\lambda(u)P_0$ is spanned by the $\C^2$-valued function 
\[z\mapsto  \e^{-|z|^2/4}\oplus 0,\] 
is proven analogously. There is a natural isomorphism 
\[\coker P_0\lambda(u)P_0 \cong (\im P_0\lambda(u)P_0)^\perp =\ker P_0\lambda(u^*)P_0.\] 
Analogously to the reasoning above, for $g\in \ker P_0\lambda(u^*)P_0$ we expand the function $g$ in an $L^2$-convergent series 
\[g=\sum_{\mathbbm{m}\in \N^2} d_\mathbbm{m}\xi^\mathbbm{m},\]
where $d_\mathbbm{m}=d_\mathbbm{m}^{(1)}\oplus d_\mathbbm{m}^{(2)}\in \C^2$. After taking scalar product by $\xi_{\mathbbm{m}'}$, for some $D_\mathbbm{m},D'_\mathbbm{m}>0$ we obtain the following conditions on the coefficients:
\begin{align}
\label{doe}
d_{\mathbbm{m}+e_1}^{(1)}&=D_\mathbbm{m}d_{\mathbbm{m}-e_2}^{(2)}\quad \mbox{and}\\ 
&d_{\mathbbm{m}+e_2}^{(1)}=-D_\mathbbm{m}'d_{\mathbbm{m}-e_1}^{(2)}.
\end{align}
The second of these equations implies $d^{(1)}_\mathbbm{m}=0$ for $m_1=0$ and $m_2>0$. Also, the first of these equations implies $d^{(1)}_\mathbbm{m}=0$ for $m_2=0$ and $m_1>0$. For $m_1,m_2>0$, putting in $\mathbbm{m}-e_1$ in the first equation, gives 
\[d_{\mathbbm{m}}^{(1)}=D_{\mathbbm{m}-e_1}d_{\mathbbm{m}-e_1-e_2}^{(2)}.\] 
Finally, combining this relation with the second equation for $\mathbbm{m}-e_2$ we obtain 
\[d^{(1)}_\mathbbm{m}\left(1+ \frac{D_{\mathbbm{m}-e_1}}{D'_{\mathbbm{m}-e_2}}\right)=0\quad \mbox{for}\quad m_1,m_2> 0.\]
Therefore $d^{(1)}_\mathbbm{m}=0$ for all $\mathbbm{m}\neq 0$. The equations in \eqref{doe} imply $d^{(2)}_\mathbbm{m}=0$ for all $\mathbbm{m}$. However, the function $z\mapsto  \e^{-|z|^2/4}\oplus 0$, corresponding to $d^{(1)}_0=1$, is in the space $\ker (P_0\lambda(u^*)P_0)$ which completes the proof.
\end{proof}

\section{The index formula on the particular Landau levels}

In this section we will prove an index formula for the particular Landau levels. On $S^{2n-1}$ we have the complex coordinates $z_1,\ldots, z_n$ and we denote by $Z_1,\ldots, Z_n$ the image of these coordinate functions under the representation $\lambda$ which was defined in equation \eqref{deflam}. So $Z_i$ is the operator on $L^2(\C^n)$ given by multiplication by the almost everywhere defined function $z\mapsto \frac{z_i}{|z|}$. Consider the polar decompositions
\[P_0Z_iP_0=V_{i,0}S_{i,0},\]
 where $V_{i,0}$ are partial isometries and $S_{i,0}> 0$. An orthonormal basis for $\ellL_0$ is given by 
\[\eta_\mathbbm{m}(z):= \frac{z^\mathbbm{m}\e^{-|z|^2/4}}{\sqrt{\pi^n 2^{|\mathbbm{m}|+n}\mathbbm{m}!}},\]
see more in \cite{follandphase}.
 
\begin{prop}
\label{firstisometry}
The operator $V_{i,0}$ is an isometry described by the equation 
\[V_{i,0}\eta_\mathbbm{m}=\eta_{\mathbbm{m}+e_i}\]
and the operator $S_{i,0}$ is diagonal in the basis $\eta_\mathbbm{m}$ with eigenvalues given by 
\begin{equation}
\label{asympsland}
\lambda_{i,\mathbbm{m}}^\eta=\Gamma\left(|\mathbbm{m}|+n+\frac{1}{2}\right)\frac{\sqrt{m_i+1}}{(|\mathbbm{m}|+n)!}.
\end{equation}
\end{prop}

\begin{proof}
For $\mathbbm{m},\mathbbm{m}'\in \N$ we have 
\begin{align*}
\langle \eta_{\mathbbm{m}'},Z_i \eta_{\mathbbm{m}}\rangle=
\int_{\C^n} \frac{1}{\pi^n\sqrt{2^{|\mathbbm{m}+\mathbbm{m}'|+2n}\mathbbm{m}!\mathbbm{m}'!}}\frac{\bar{z}^{\mathbbm{m}'}z^{\mathbbm{m}+e_i}}{|z|}\e^{-|z|^2/2}\rd &V=\\
=\frac{1}{\pi^n\sqrt{2^{|\mathbbm{m}+\mathbbm{m}'|+2n}\mathbbm{m}!\mathbbm{m}'!}}\int_0^\infty r^{|\mathbbm{m}|+|\mathbbm{m}'|+n-1}\e^{-r^2/2}\rd r\int_{S^{2n-1}}&\bar{z}^{\mathbbm{m}'}z^{\mathbbm{m}+e_i}\rd S=\\
=\delta_{\mathbbm{m}',\mathbbm{m}+e_i}\frac{\Gamma\left(|\mathbbm{m}|+n+\frac{1}{2}\right)}{2\pi^n\mathbbm{m}!\sqrt{(\mathbbm{m}_j+1)}}\int_{S^{2n-1}}&\bar{z}^{\mathbbm{m}'}z^{\mathbbm{m}+e_i}\rd S=\\
=\delta_{\mathbbm{m}',\mathbbm{m}+e_i}\Gamma(|\mathbbm{m}|+n+&\frac{1}{2})\frac{\sqrt{m_i+1}}{(|\mathbbm{m}|+n)!}.
\end{align*}
It follows that $V_{i,0}\eta_\mathbbm{m}=\eta_{\mathbbm{m}+e_i}$ and $S_{i,0}\eta_\mathbbm{m}=\lambda_{i,\mathbbm{m}}^\eta \eta_\mathbbm{m}$, where $\lambda_{i,\mathbbm{m}}^\eta$ is as in equation \eqref{asympsland}.
\end{proof}

On the other hand, we can, just as on $\ellL_0$, let  $\tilde{Z}_1,\ldots, \tilde{Z}_n\in \Bo(L^2(B_n))$ be the operators on $L^2(B_n)$ defined by the multiplication by the almost everywhere defined function $z\mapsto \frac{z_i}{|z|}$. Consider the polar decompositions
\[P_B\tilde{Z}_iP_B=V_{i,B}S_{i,B},\]
 where again $V_{i,B}$ are partial isometries and $S_{i,B}> 0$. An orthonormal basis for $A^2(B_n)$ is given by 
\[\mu_\mathbbm{m}(z):= \pi^{-n/2}\sqrt{\frac{(n+|\mathbbm{m}|)!}{\mathbbm{m}!}}z^\mathbbm{m}.\]
Similar to the lowest Landau level, the partial isometries $V_{i,B}$ are just shifts in this basis:

\begin{prop}
\label{secondisometry}
The operator $V_{i,B}$ is an isometry described by the equation 
\[V_{i,B}\mu_\mathbbm{m}=\mu_{\mathbbm{m}+e_i}\]
and the operator $S_{i,B}$ is diagonal in the basis $\mu_\mathbbm{m}$ with eigenvalues given by 
\begin{equation}
\label{asympsberg}
\lambda_{i,\mathbbm{m}}^\mu=\frac{\sqrt{m_i+1}}{\sqrt{n+|\mathbbm{m}|+1}}.
\end{equation}
\end{prop}

\begin{proof}
The proof is the analogous to that of Proposition \ref{firstisometry}. For $\mathbbm{m},\mathbbm{m}'\in \N$ we have 
\begin{align*}
\langle \mu_{\mathbbm{m}'},\tilde{Z}_i \mu_{\mathbbm{m}}\rangle=\int_{B_n} \pi^{-n}\sqrt{\frac{(n+|\mathbbm{m}|)!(n+|\mathbbm{m}'|)!}{\mathbbm{m}!\mathbbm{m}'!}}\frac{\bar{z}^{\mathbbm{m}'}z^{\mathbbm{m}+e_i}}{|z|}\rd &V=\\
= \pi^{-n}\sqrt{\frac{(n+|\mathbbm{m}|)!(n+|\mathbbm{m}'|)!}{\mathbbm{m}!\mathbbm{m}'!}}\int_0^1	 r^{|\mathbbm{m}|+|\mathbbm{m}'|+2n-1}\rd r\int_{S^{2n-1}}&\bar{z}^{\mathbbm{m}'}z^{\mathbbm{m}+e_i}\rd S=\\
=\delta_{\mathbbm{m}',\mathbbm{m}+e_i}\frac{(n+|\mathbbm{m}|)!\sqrt{n+|\mathbbm{m}|+1}}{(2|\mathbbm{m}|+2n)\mathbbm{m}!\sqrt{m_i+1}}\int_{S^{2n-1}}&\bar{z}^{\mathbbm{m}'}z^{\mathbbm{m}+e_i}\rd S=\\
=\delta_{\mathbbm{m}',\mathbbm{m}+e_i}&\frac{\sqrt{m_i+1}}{\sqrt{n+|\mathbbm{m}|+1}}.
\end{align*}
It follows that $V_{i,B}\mu_\mathbbm{m}=\mu_{\mathbbm{m}+e_i}$ and $S_{i,B}\mu_\mathbbm{m}=\lambda_{i,\mathbbm{m}}^\mu \mu_\mathbbm{m}$ where the eigenvalues $\lambda_{i,\mathbbm{m}}^\mu$ are given in equation \eqref{asympsberg}.
\end{proof}

\begin{lem}
\label{gammaestimate}
If $a$ is a real number then
\[\frac{\Gamma(x+a)}{\Gamma(x)}= x^a+O(x^{-1+a}) \quad \mbox{as} \quad x\to +\infty.\]
\end{lem}

\begin{proof}
By Stirling's formula
\[\ln \Gamma (x) = \left(x-\frac12\right)\ln x -x + \frac{\ln {2 \pi}}{2} + O(x^{-1}).\]
After Taylor expanding $\ln \Gamma(x+a)$ around $a=0$ we obtain that 
\[\ln\Gamma(x+a)-\ln\Gamma(x)=a\ln x+O(x^{-1}).\]
\end{proof}

\begin{lem}
\label{unitcomp}
With the unitary $U:A^2(B_n)\to \ellL_0$ defined by $\mu_\mathbbm{m}\mapsto \eta_\mathbbm{m}$, the operators $S_{i,0}$ and $S_{i,B}$ satisfy 
\[U^*S_{i,0}U-S_{i,B}\in \Ko.\]
\end{lem}

\begin{proof}
The operators $U^*S_{i,0}U$ and $S_{i,B}$ are both diagonal in the basis $\mu_\mathbbm{m}$. So it is sufficient to prove that $|\lambda_\mathbbm{m}^\eta-\lambda_\mathbbm{m}^\mu|\to 0$. The proof of this statement is based on the estimate from Lemma \ref{gammaestimate}. When $|\mathbbm{m}|\to \infty$, Lemma \ref{gammaestimate} implies   
\begin{align*}
|\lambda_\mathbbm{m}^\eta-\lambda_\mathbbm{m}^\mu|&=\left|\frac{\Gamma\left(|\mathbbm{m}|+n+\frac{1}{2}\right)\sqrt{m_i+1}}{(|\mathbbm{m}|+n)!}-\frac{\sqrt{m_i+1}}{\sqrt{|\mathbbm{m}|+n-1}}\right|=\\
&=\sqrt{m_i+1}\left|\frac{\Gamma\left((|\mathbbm{m}|+n+1)-\frac{1}{2}\right)}{\Gamma\left(|\mathbbm{m}|+n+1\right)}-(|\mathbbm{m}|+n-1)^{-1/2}\right|= O(|\mathbbm{m}|^{-1}).
\end{align*}
Therefore we have that $U^*S_{i,0}U-S_{i,B}\in \ellL^{n+}(A^2(B_n))$, the $n$:th Dixmier ideal. In particular $U^*S_{i,0}U-S_{i,B}$ is compact.
\end{proof}

\begin{sats}
\label{isot}
The unitary $U$ induces an isomorphism $Ad(U):\mathcal{T}_0\xrightarrow{\sim}\mathcal{T}^n$ such that 
\[\sigma^n\circ Ad(U)=\sigma_0.\]
where $\sigma^n$ and $\sigma_0$ are the symbol mappings.
\end{sats}

\begin{proof}
Lemma \ref{unitcomp} and the Propositions \ref{firstisometry} and \ref{secondisometry} imply
\begin{equation}
\label{precomm}
U^*(P_0Z_iP_0)U=P_B\tilde{Z}_iP_B+K_i,
\end{equation}
 for some compact operators $K_i$. Since $\mathcal{T}^n$ contains the compact operators, $U^*(P_0Z_iP_0)U\in \mathcal{T}^n$. Corollary \ref{genlem} therefore implies $U^*\mathcal{T}_0U \subseteq \mathcal{T}^n$. Theorem \ref{genlam} states that $\mathcal{T}_0$ acts irreducibly on $\ellL_0$, so $U^*\mathcal{T}_0U$ acts irreducibly on $A^2(B_n)$. Therefore $\Ko\subseteq U^*\mathcal{T}_0U$ and $P_B\tilde{Z}_iP_B\in U^*\mathcal{T}_0U$. The operators $P_B\tilde{Z}_iP_B$ together with $\Ko$ generate $\mathcal{T}^n$ so $U^*\mathcal{T}_0U\supseteq \mathcal{T}^n$. The relation $\sigma^n\circ Ad(U)=\sigma_0$ holds since by equation \eqref{precomm} it holds on the generators of $C(S^{2n-1})$.
\end{proof}

\begin{cor}
\label{indextheorem}
Let $[\mathcal{T}^n]\in Ext(C(S^{2n-1}))$ denote the Toeplitz quantization of the Bergman space defined in equation \eqref{to} and $[\mathcal{T}_\mathbbm{k}]\in Ext(C(S^{2n-1}))$ the Toeplitz quantization of the particular Landau level of height $\mathbbm{k}$ defined in equation \eqref{lato}. Then 
\[[\mathcal{T}^n]=[\mathcal{T}_\mathbbm{k}].\]
So for $u\in A\otimes M_N$ such that $u_\partial:=\pi_\partial(u)$ is invertible and smooth
\begin{equation}
\label{indeform}
\ind (P_\mathbbm{k}u|_{\ellL_\mathbbm{k}\otimes \C^N})=\frac{-(n-1)!}{(2n-1)!(2\pi i)^n}\int _{S^{2n-1}}\tra((u_\partial^{-1}\rd u_\partial)^{2n-1}).
\end{equation}
\end{cor}

\begin{proof}
By Corollary \ref{indepland} the class $[\mathcal{T}_\mathbbm{k}]$ is independent of $\mathbbm{k}$, so take $\mathbbm{k}=0$. In this case Theorem \ref{isot} implies that the unitary $U$ makes the following diagram commutative:
\[
\begin{CD}
0@>>>   \Ko@>>>\mathcal{T}_{0}@>\sigma_0>> C(S^{2n-1})@>>>0\\
@. @VVAd(U)V @VV Ad(U)V @|@.\\
0@>>>   \Ko@>>>\mathcal{T}^n@>\sigma^n>> C(S^{2n-1})@>>>0\\
\end{CD}.\]
Therefore $[\mathcal{T}^n]=[\mathcal{T}_0]=[\mathcal{T}_\mathbbm{k}]$ and the index formula \eqref{indeform} follows from \cite{guehig}.
\end{proof}


\newpage




\begin{thebibliography}{99}

\bibitem{avronseilersimon} J. E. Avron, R. Seiler, B. Simon, \emph{Charge deficiency, charge transport and comparison of dimensions},  Comm. Math. Phys.  159  (1994),  no. 2, 399--422. 

\bibitem{coburnberger} C. A. Berger, L. A. Coburn, \emph{Toeplitz operators on the Segal-Bargmann space}, Trans. Amer. Math. Soc.  301  (1987),  no. 2, 813--829.

\bibitem{boutetdemonvel} L. Boutet de Monvel, \emph{On the index of Toeplitz operators of several complex variables},  Invent. Math.  50  (1978/79), no. 3, 249--272.

\bibitem{bmrs} J. Brodzki, V. Mathai, J. Rosenberg, R.J. Szabo, \emph{$D$-Branes, $RR$-Fields and Duality on Noncommutative Manifolds}, arXiv:hep-th/0607020v3.

\bibitem{bpr} V. Bruneau, A. Pushnitski, G. Raikov, \emph{Spectral shift function in strong magnetic fields}, St. Petersburg Math J., Vol. 16 (2005), No. 1, p. 181-209.



\bibitem{doug} R. G. Douglas, \emph{Banach algebra techniques in the theory of Toeplitz operators}, AMS 1972.

\bibitem{follandphase} G. B. Folland, \emph{Harmonic analysis in phase space}, Annals of Mathematics Studies, 122. Princeton University Press, Princeton, NJ, 1989. x+277 pp.

\bibitem{guehig} E. Guentner, N. Higson, \emph{A note on Toeplitz operators}, Internat. J. Math.  7  (1996),  no. 4, 501--513.

\bibitem{melroz} M. Melgaard, G. Rozenblum,  \emph{Schr\"odinger operators with singular potentials} , Chapter 6 in  Stationary Partial Differential Equations , pp 407-517, Volume II of the multivolume HANDBOOK OF DIFFERENTIAL EQUATIONS, M. Chipot and P. Quittner (Edit.), Elsevier B.V. 2005. 

\bibitem{reisszabo} R. M. G. Reis, R. J. Szabo, \emph{Geometric K-Homology of Flat D-Branes}, arXiv:hep-th/0507043v3.


\bibitem{rosenbergschochet} J. Rosenberg, C. Schochet, \emph{The classification of extensions of $C\sp{\ast} $-algebras},  Bull. Amer. Math. Soc. (N.S.)  4  (1981), no. 1, 105--110. 

\bibitem{grigsobolev}  G. Rozenblum, A. V. Sobolev, \emph{Discrete spectrum distribution of the Landau Operator Perturbed by an Expanding Electric Potential}, arXiv:0711.2158

\bibitem{rozta} G. Rozenblum, G. Tashchiyan, \emph{On the spectral properties of the perturbed Landau Hamiltonian},  Comm. Partial Differential Equations  33  (2008),  no. 4-6, 1048--1081. 

\bibitem{venugopalkrishna} U. Venugopalkrishna, \emph{Fredholm operators associated with strongly pseudoconvex domains in $\C\sp{n}$}, J. Functional Analysis 9 (1972), 349--373. 

\end{thebibliography}
\end{document}